\documentclass[submission,copyright,creativecommons]{eptcs}
\providecommand{\event}{SOS 2007} 
\usepackage{underscore}           

\usepackage{bbm}
\usepackage{amssymb,amsfonts}
\usepackage{stmaryrd} 

\usepackage{relsize}
\usepackage{graphicx}

\usepackage[font=small]{caption}

\usepackage[version=0.96]{pgf}
\usepackage{tikz}
\usetikzlibrary{arrows,positioning,shapes,decorations,automata,backgrounds,petri}

\usepackage{hyperref} 

\DeclareMathAlphabet{\mathpzc}{OT1}{pzc}{m}{it}
\tikzset{fontscale/.style = {font=\relsize{#1}}}

\title{Regular $\omega$-Languages with an Informative Right Congruence}
\author{Dana Angluin
\institute{Yale University}
\email{angluin@cs.yale.edu}
\and
Dana Fisman 
\institute{Ben-Gurion University}
\email{dana@cs.bgu.ac.il}
}
\def\titlerunning{Regular $\omega$-Languages with an Informative Right Congruence}
\def\authorrunning{D. Angluin and D. Fisman}
\newtheorem{theorem}{Theorem}
\newtheorem{proposition}[theorem]{Proposition}
\newtheorem{corollary}[theorem]{Corollary}
\newtheorem{definition}[theorem]{Definition}
\newtheorem{lemma}[theorem]{Lemma}
\newenvironment{proof}{\paragraph{Proof.}}{\hfill$\square$ \\}

\usepackage{pifont}
\newcommand{\cmark}{\ding{51}}%
\newcommand{\xmark}{\ding{55}}%

\newcommand\Red[1]{\textcolor[rgb]{.6875,0,0}{#1}}

\newcommand{\commentout}[1]{}
\newcommand{\tbd}[1]{{\Red{\textsl{\textbf{TBD: }{#1}}}}}

\newcommand{\la}{\langle}
\newcommand{\ra}{\rangle}

\newcommand{\Lang}[1]{{\llbracket #1 \rrbracket}}
\newcommand{\sema}[1]{{\Lang{#1}}}

\newcommand{\fdfa}{{\textsc{fdfa}}}

\newcommand{\scc}{{\textsc{scc}}}
\newcommand{\dfa}{{\textsc{dfa}}}

\newcommand{\dba}{{\textsc{dba}}}
\newcommand{\nba}{{\textsc{nba}}}
\newcommand{\dca}{{\textsc{dca}}}
\newcommand{\nca}{{\textsc{nca}}}
\newcommand{\dma}{{\textsc{dma}}}
\newcommand{\dta}{{\textsc{dta}}}
\newcommand{\dpa}{{\textsc{dpa}}}

\newcommand{\nma}{{\textsc{nma}}}
\newcommand{\nta}{{\textsc{nta}}}
\newcommand{\npa}{{\textsc{npa}}}

\newcommand{\buchi}{B$\ddot{\textrm{u}}$chi}
\newcommand{\muller}{Muller}
\newcommand{\tmuller}{Transition-Table}

\newcommand{\cobuchi}{co-B$\ddot{\textrm{u}}$chi}

\newcommand{\DB}{\ensuremath{\mathbb{DB}}}
\newcommand{\DC}{\ensuremath{\mathbb{DC}}}

\newcommand{\DM}{\ensuremath{\mathbb{DM}}}
\newcommand{\NB}{\ensuremath{\mathbb{NB}}}
\newcommand{\NC}{\ensuremath{\mathbb{NC}}}
\newcommand{\NP}{\ensuremath{\mathbb{NP}}}
\newcommand{\NM}{\ensuremath{\mathbb{NM}}}
\newcommand{\NT}{\ensuremath{\mathbb{NT}}}

\newcommand{\DP}{\ensuremath{\mathbb{DP}}}
\newcommand{\DT}{\ensuremath{\mathbb{DT}}}

\newcommand{\NN}{\mathbb{N}}

\newcommand{\class}[1]{\ensuremath{\mathbb{#1}}}
\newcommand{\IM}{\ensuremath{\mathbb{IM}}}
\newcommand{\IT}{\ensuremath{\mathbb{IT}}}
\newcommand{\IB}{\ensuremath{\mathbb{IB}}}

\newcommand{\IP}{\ensuremath{\mathbb{IP}}}
\newcommand{\IC}{\ensuremath{\mathbb{IC}}}

\newcommand{\RT}{\ensuremath{\mathbb{RT}}}
\newcommand{\RM}{\ensuremath{\mathbb{RM}}}
\newcommand{\RB}{\ensuremath{\mathbb{RB}}}
\newcommand{\RP}{\ensuremath{\mathbb{RP}}}
\newcommand{\RC}{\ensuremath{\mathbb{RC}}}

\newcommand{\lstar}{\ensuremath{{L^*}}}

\newcommand{\aut}[1]{{\mathpzc{#1}}}

\newcommand{\A}{{\aut{A}}}

\newcommand{\F}{{{\aut{F}}}}

\newcommand{\R}{{\aut{R}}}
\newcommand{\B}{{{\aut{B}}}}
\newcommand{\C}{{{\aut{C}}}}
\newcommand{\D}{{{\aut{D}}}}
\newcommand{\N}{{{\aut{N}}}}
\newcommand{\Q}{{{\aut{Q}}}}
\newcommand{\M}{{\aut{M}}}

\newcommand{\T}{{{\aut{T}}}}

\renewcommand{\P}{{{\aut{P}}}}
\renewcommand{\T}{{{\aut{T}}}}

\newcommand{\rightaut}{rightcon}
\newcommand{\initstate}{\lambda}

\newcommand{\naturals}{\NN}

\begin{document}
\maketitle

\begin{abstract}
	A regular language is almost fully characterized by its right congruence relation. Indeed, a regular language can always be recognized by a \dfa\ isomorphic to the automaton corresponding to its right congruence, henceforth the \emph{rightcon automaton}. The same does not hold for regular $\omega$-languages. The right congruence of a regular  $\omega$-language is not informative enough; many regular $\omega$-languages have a trivial right congruence, and in general it is not always possible to define an $\omega$-automaton recognizing a given language that is isomorphic to the rightcon automaton. 

	The class of \emph{weak regular $\omega$-languages} does have an informative right congruence. That is, any weak regular $\omega$-language can always be recognized by a deterministic \buchi\ automaton that is isomorphic to the rightcon automaton. 
	Weak regular $\omega$-languages reside in the lower levels of the expressiveness hierarchy of regular $\omega$-languages. Are there more expressive sub-classes of regular $\omega$-languages that have an informative right congruence? 
	Can we fully characterize the class of languages with a trivial right congruence? In this paper we try to place some 
	additional pieces of this big puzzle.

\end{abstract}

\section{Introduction}
	Regular $\omega$-languages play a key role in reasoning about reactive systems. Algorithms for verification and synthesis of reactive system typically build on the theory of $\omega$--automata. The theory of $\omega$-automata enjoys many properties that the theory of automata on finite words enjoys. These make it amenable for providing the basis for analysis algorithms (e.g. model checking algorithms rely on the fact that emptiness can be checked in nondeterministic logarithmic space). However, in general, the theory of $\omega$-automata is much more involved than that of automata on finite words, and many fundamental questions, such as minimization, are still open. 

One of the fundamental theorems of regular languages on finite words is the Myhill-Nerode theorem stating a one-to-one correspondence between the state of the minimal deterministic finite automaton (\dfa) for a language $L$ and the equivalence classes of the right congruence of $L$.\footnote{Formal definitions are deferred to Section~\ref{sec:prelim}.}  When moving to $\omega$-words, there is no similar theorem, and there are many regular $\omega$-languages where any minimal automaton requires more states than the number of equivalence classes in the right congruence of the language. For instance, consider the $\omega$-language $L=(a+b)^*a^\omega$. Its right congruence has only one equivalence class. That is, for any finite  words $x$ and $y$ and any $\omega$-word $w$ we have that $xw\in L$ iff $yw \in L$ as membership in $L$ is determined only by the suffix. We say that the right congruence for $L$ is not \emph{informative enough}.

The tight relationship between the equivalence classes of the right congruence and the states of a minimal \dfa\ is at the heart of minimization and learning algorithms for regular languages of finite words, and seems to be a severe obstacle in obtaining efficient minimization and learning algorithms for regular $\omega$-languages. For this reason, we set ourselves to study classes of regular $\omega$-languages that do have a right congruence that is fully informative. 

Several acceptance criteria are in use for $\omega$-automata, in particular, \buchi, \cobuchi, Muller and Parity. There are differences in the expressiveness of the corresponding deterministic automata. We use $\DB$, $\DC$, $\DP$ and $\DM$ to denote the classes of languages accepted by deterministic \buchi, \cobuchi, Muller and Parity automata, respectively. We also consider a version of Muller automata where acceptance is defined using transitions, refer to this acceptance criterion as Transition-Table, and use $\DT$ for the corresponding class of languages. The classes $\DT$, $\DM$ and $\DP$ accept all regular $\omega$-languages whereas $\DB$ and $\DC$ are strictly less expressive and are dual to each other. The intersection of $\DB$ and $\DC$ is the class of \emph{weak} regular $\omega$-languages. This class does have the property that any language in $\DB\cap \DC$ has a fully informative right congruence.  The regular $\omega$-languages can be arranged in an infinite hierarchy of expressive power as suggested by Wagner~\cite{Wagner75} and the class $\DB\cap \DC$ corresponds to one of the lowest levels of the hierarchy.

We define the classes \class{IB}, \class{IC}, \class{IP}, \class{IM}, \class{IT} to be the class of regular $\omega$-languages that can be accepted by a \buchi, \cobuchi, Parity and Muller and Transition-Table automata, respectively, whose number of states equals the number of equivalence classes in the right congruence of the language. We show that these form a strictly inclusive hierarchy of expressiveness as shown in Fig.~\ref{fig:ibicimip-relations} on the left, and moreover in every class of the infinite Wagner hierarchy of regular $\omega$-languages there are languages whose right congruence is fully informative. 

Noting another difficulty in inferring a regular $\omega$-language from examples of $\omega$-words in the language, we consider a further restriction on languages, that if $ux^\omega$ is accepted by a minimal automaton for the language, then that automaton has a loop of size at most $|x|$ in which $ux^\omega$ loops. We term this property, being \emph{respective} of the right congruence.
This property is reminiscent of the property of being non-counting~\cite{DiekertG08}, and we show that a language that is non-counting is respective of its right congruence but the other direction does not necessarily hold.
We define the classes  \class{RB}, \class{RC}, \class{RP}, \class{RM}, \class{RT} of languages in \class{IB}, \class{IC}, \class{IP}, \class{IM}, \class{IT} that are respective of their right congruence. We show that these classes constitute a further restriction in terms of expressive power as shown in Fig.~\ref{fig:RXinclusions} on the right, and yet here as well, in every class of the infinite Wagner hierarchy of regular $\omega$-languages there are languages which are respective of their fully informative right congruence.

\begin{figure}[h]
	\noindent\makebox[\textwidth]{
		\begin{tabular}{l@{\qquad\qquad}r}
			\scalebox{0.6}{
				\def\classib{(-0.6,0) ellipse (4em and 2.5em)}
				\def\classic{(0.6,0) ellipse (4em and 2.5em)}
				\def\classip{(0,0.5) ellipse (7em and 4.5em)}
				\def\classim{(0,0.95) ellipse (10em and 6.3em)}
				\def\classit{(0,1.4) ellipse (13em and 8.1em)}

				\begin{tikzpicture}
				\draw \classib node [below] { };
				\draw \classic node [below] { };
				\draw \classip node [above] { };
				\draw \classim node [above] { };
				\draw \classit node [above] { };				
				
				\node [align=center] at (0,0) (B) {\DB \\[.3em] $\cap$ \\[.3em] \DC};
				\node at (-1.4,0) (B) {\IB};
				\node at (1.4,0) (C) {\IC};
				\node at (0,1.5) (P) {\IP};
				\node at (0,2.6) (M) {\IM};
				\node at (0,3.7) (T) {\IT};				
				\end{tikzpicture}
			}
			& 
			\scalebox{0.6}{
				\def\classib{(-0.6,0) ellipse (4em and 2.5em)}
				\def\classrb{(-0.45,0) ellipse (2.9em and 2.25em)}
				\def\classic{(0.6,0) ellipse (4em and 2.5em)}
				\def\classrc{(0.45,0) ellipse (2.9em and 2.25em)}
				\def\classrp{(0,0.222) ellipse (6.3em and 3.65em)}
				\def\classip{(0,0.5) ellipse (7em and 4.5em)}
				\def\classrm{(0,0.75) ellipse (8.5em and 5.4em)}
				\def\classim{(0,0.95) ellipse (10em and 6.3em)}
				\def\classrt{(0,1.15) ellipse (11.5em and 7.4em)}					
				\def\classit{(0,1.35) ellipse (13em and 8.6em)}

				\begin{tikzpicture}
				\draw \classib node [below] { };
				\draw \classic node [below] { };
				\draw \classip node [above] { };
				\draw \classim node [above] { };
				\draw \classit node [above] { };
				\draw \classrt node [above] { };										
				\draw \classrm node [above] { };
				\draw \classrp node [above] { };
				\draw \classrc node [above] { };
				\draw \classrb node [above] { };

				\node at (-1.85,0) (B) {\IB};
				\node at (-1.25,0) (RB) {\RB};
				
				\node at (1.85,0) (C) {\IC};
				\node at (1.25,0) (RC) {\RC};
				
				\node at (0,2.5) (P) {\IP};
				\node at (0,1.2) (RP) {\RP};
				
				\node at (0,3.1) (M) {\IM};
				\node at (0,1.9) (RM) {\RM};
				
				\node at (0,4.35) (T) {\IT};
				\node at (0,3.7) (RT) {\RT};
				
				\end{tikzpicture}
			}						
		\end{tabular}
	}
	\caption{On the left, a summary of inclusion of the classes of $\omega$-automata that are isomorphic to their rightcon automaton. On the right a summary of inclusion of the classes of $\omega$-automata that are isomorphic to their rightcon automaton, as well as those that are in addition respective of their right congruence.}
	\label{fig:ibicimip-relations}
	\label{fig:RXinclusions}
\end{figure}
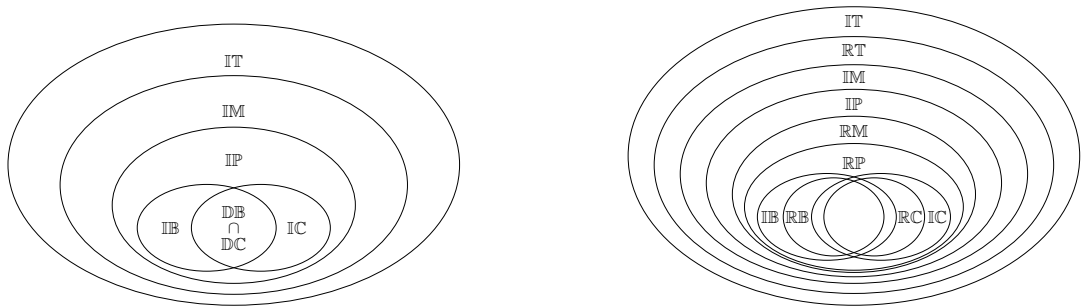

The rest of the paper is organized as follows. In Section~\ref{sec:prelim} we provide the necessary definitions of the $\omega$-automata that we consider and of right congruence,  state the well known results about their expressiveness, and briefly summarize the importance of the relation between the  syntactic right congruence of a language and its minimal acceptor in learning algorithms. In Section~\ref{sec:relations-to-the-right-aut} we state the relations between states of arbitrary $\omega$-automata for a regular $\omega$-language $L$ and its syntactic right congruence. In Section~\ref{sec:trivial} we provide a full characterization of $\omega$-languages $L$ for which $\sim_L$ is trivial. In Section~\ref{sec:isomorphic} we explore expressiveness results related to the classes of languages with an informative right congruence. In Section~\ref{sec:respective} we explore expressiveness results related to the classes of languages that not only have an informative right congruence, but are also respective of their right congruence. In Section~\ref{subsec:closure} we explore closure properties of these classes, and in Section~\ref{sec:discussion} we conclude.
Due to lack of space some proofs are missing, these can be found in the full version of the paper at http://www.cis.upenn.edu/$\sim$fisman/documents/AF_GandALF18_full.pdf

\section{Preliminaries}
\label{sec:prelim}
An \emph{alphabet} $\Sigma$ is a finite set of symbols. The set of finite words over $\Sigma$ is denoted by $\Sigma^*$, and the set of infinite words, termed $\omega$-words, over $\Sigma$ is denoted by $\Sigma^\omega$. We use $\epsilon$ for the empty word and $\Sigma^+$ for $\Sigma^*\setminus\{\epsilon\}$.  A \emph{language} is a set of finite words, that is a subset of $\Sigma^*$, while an $\omega$-language is a set of $\omega$-words, that is a subset of $\Sigma^\omega$. For natural numbers $i$ and $j$ and a word $w$, we use $[i..j]$ for the set $\{i, i+1, \ldots, j\}$, $w[i]$ for the $i$-th letter of $w$, and $w[i..j]$ for the subword of $w$ starting at the $i$-th letter and ending at the $j$-th letter, inclusive.

\paragraph*{Automata and Acceptors}
An \emph{automaton} is a tuple $\A=\la \Sigma, Q, \initstate, \delta \ra$ consisting of an alphabet $\Sigma$, a finite set $Q$ of
states, an initial state $\initstate\in Q$, and a transition function ${\delta: Q \times \Sigma \rightarrow 2^Q}$. 
A run of an automaton on a finite word ${v=a_1 a_2\ldots a_n}$ is a sequence of states ${q_0,q_1,\ldots,q_n}$ such that $q_0=\initstate$, and for each $i\geq 0$, ${q_{i+1}\in\delta(q_i,a_{i+1})}$.  A run on an infinite word is defined similarly and results in an infinite sequence of states.
The transition function is naturally extended to a function from $Q \times \Sigma^*$, by defining $\delta(q,\epsilon)=q$ and ${\delta(q,a v)=\delta(\delta(q,a),v)}$ for ${q\in Q}$, ${a\in\Sigma}$ and ${v\in\Sigma^*}$.
We often use $\A(v)$ as a shorthand for $\delta(\initstate,v)$ and $|\A|$ for the number of states in $Q$.
We use $\A_q$ to denote the automaton $\la \Sigma, Q, q, \delta \ra$ obtained from $\A$ by replacing the initial state with $q$. We say that $\A$ is \emph{deterministic} if $|\delta(q,a)|\leq1$ and \emph{complete} if $|\delta(q,a)|\geq1$, for every $q\in Q$ and $a\in\Sigma$. 

By augmenting an automaton with an acceptance condition $\alpha$, obtaining a tuple $\la \Sigma, Q, \initstate,$ $\delta, \alpha \ra$, we get an \emph{acceptor}, a machine that accepts some words and rejects others. An acceptor 
accepts a word if at least one of the runs on that word is accepting. For finite words the acceptance condition is a set $F \subseteq Q$ and a run on a word $v$ is accepting if it ends in an accepting state, i.e., if $\delta(\initstate,v)$ contains an element of $F$. For infinite words, there are various acceptance conditions in the literature; we consider five: \buchi, \cobuchi, parity, Muller and \tmuller~\cite{Saec90}. 
The \buchi\ and \cobuchi\ acceptance conditions are also a set $F \subseteq Q$. A run of a \buchi\ automaton is accepting if it visits $F$ infinitely often. A run of a co-\buchi\ is accepting if it visits $F$ only finitely many times. A parity acceptance condition is a map $\kappa:Q \rightarrow [1..k]$ (for some $k\in\naturals$) assigning each state a color (or priority). A run is accepting if the minimal color visited infinitely often is odd. A \muller\ acceptance condition is a set of sets of states $\alpha=\{F_1,F_2,\ldots,F_k\}$ for some $k\in\naturals$ and $F_i\subseteq Q$ for $i\in[1..k]$. A run of a \muller\ automaton is accepting if the set $S$ of states visited infinitely often in the run is in $\alpha$. A \tmuller\ acceptance condition is a set $\alpha=\{T_1,T_2,\ldots,T_k\}$ of sets of transitions, where a transition is a tuple in $Q\times \Sigma \times Q$. A run of a \tmuller\ automaton is accepting if the set $T$ of transitions visited infinitely often in the run is in $\alpha$. 
We use $\sema{\A}$ to denote the set of words accepted by a given acceptor $\A$. 

%
\begin{figure}[h]
	\begin{center}
		\noindent\makebox[\textwidth]{
			\scalebox{.6}{
				\begin{tikzpicture}[->,>=stealth',shorten >=1pt,auto,node distance=2.2cm,semithick,initial text=,initial where=left]
				
				\node[label]          (L)                {};
				\node[label]          (UL) [above of=L, node distance=2.6cm, fontscale=1.5]                {$\aut{B}:$};
				
				\node[state]          (A)   [right of=L] {${a}$};
				\node[initial,state,accepting](AB)    [below left  of=A]{${ab}$};
				\node[state]  (E)   [below right of=A]{${\epsilon}$};

				\path (E) edge [right]  node {\textbf{${a}$}} (A); 
				\path (E) edge [loop right] node {\textbf{$\Sigma\setminus\{a\}$}} (E); 
				\path (A) edge [loop above] node {\textbf{$\Sigma\setminus\{b\}$}} (A); 
				\path (A) edge  [left]          node {\textbf{$b$}} (AB); 
				\path (AB) edge [below] node {\textbf{$\Sigma$}} (E);
				
				\node[label]  (I) [above right of=E] {};
				\node[label]          (LL)   [ right of=I, node distance=1cm]            {};
				\node[label] (LLL) [right of=UL, node distance=7cm, fontscale=1.5] {$\aut{C}:$}; 
				\node[initial,state,accepting]          (Q0)   [below right of=LL] {${1}$};
				\node[state] (Q1)    [above  of=Q0]{${2}$};
				
				\path (Q0) edge [loop below] node {$a$} (Q0); 
				\path (Q0) edge  [bend left]          node {$b$} (Q1); 
				\path (Q1) edge [loop above] node {$b$} (Q1); 
				\path (Q1) edge  [bend left]          node {$a$} (Q0);

				\node[label]          (xLL)   [ right of=I, node distance=5cm]            {};
				\node[label] (xLLL) [right of=UL, node distance=11cm, fontscale=1.5] {$\aut{M}:$};
				
				\node[label] (xII) [right of=xLLL, node distance=1.5cm, fontscale=1.5] {$\quad\{ \{1\}, \{2\}\}$};
				\node[initial,state]          (xQ0)   [below right of=xLL] {${1}$};
				\node[state] (xQ1)    [above  of=xQ0]{${2}$};
				
				\path (xQ0) edge [loop below] node {$b$} (xQ0); 
				\path (xQ0) edge  [bend left]          node {$a$} (xQ1); 
				\path (xQ1) edge [loop above] node {$b$} (xQ1); 
				\path (xQ1) edge  [bend left]          node {$a$} (xQ0);

				\node[label]            (Aut)  [right of=xLL, node distance=4cm]             {};
				\node[label]            (AutI)  [above of=Aut]             {};
				
				\node[label]            (PAut)  [right of=UL, node distance=16cm, fontscale=1.5]             {$\P:$};
				\node[label]            (c1)  [right of=PAut, node distance=1.5cm, fontscale=1.5]             {$\kappa(1)=2$};
				\node[label]            (c2)  [below of=c1, node distance=0.6cm, fontscale=1.5]             {$\kappa(2)=1$};
				\node[label]            (c3)  [below of=c2, node distance=0.6cm, fontscale=1.5]             {$\kappa(3)=0$};
				\node[label]            (c4)  [below of=c3, node distance=0.6cm, fontscale=1.5]             {$\kappa(4)=0$};

				\node[initial, state] (A)  [below right of=Aut] {$1$};
				\node[state]          (B)  [right of=A, node distance=1.5cm]   {$2$};
				\node[state]          (C)  [right of=B, node distance=1.5cm]   {$3$};
				\node[state]          (D)  [right of=C, node distance=1.5cm]   {$4$};
				
				\path (A) edge [loop below]          node {$a$} (A)
				edge                       node {$b$} (B)
				(B) edge [bend left=45,above]  node {$a$} (A)
				edge                       node {$b$} (C)
				(C) edge  [bend left=62,above] node {$a$} (A)
				edge                       node {$b$} (D)
				(D) edge [loop above]          node {$\Sigma$} (D);

				\node[label] (TAut)   [right of=UL, node distance=23cm, fontscale=1.5]  {$\aut{T}:$}; 
				\node[label] (TAutL)  [below of=TAut, node distance=.75cm] {$\{ \{(1,a,1) \} \}$};
				\node[initial, state] (tA)  [right of=D, node distance=2.5cm] {$1$};
				
				\path (tA) edge [loop below]          node {$a$} (tA)
				edge [loop above]          node {$b$} (tA);

				\end{tikzpicture}}}
	\end{center}
	\vspace{-3mm}
	\caption{A \dba\ $\B$ accepting $(\Sigma^*a\Sigma^*b)^\omega$ where $\Sigma=\{a,b,c\}$, a \dca\ $\C$ accepting $(a+b)^*b^\omega$, a \dma\ $\M$ accepting $(a+b)^*b^\omega$, a \dpa\ $\P$ accepting $(a + ba + bba)^*(a^*ba)^\omega$, and a \dta\ $\T$ accepting $(a+b)^*a^\omega$.}
	\label{fig:w-aut-examples}
	\label{fig:minimal-buchi-aut-ex}
	\label{fig:minimal-muller-aut-ex}
	\label{fig:dta}
\end{figure}
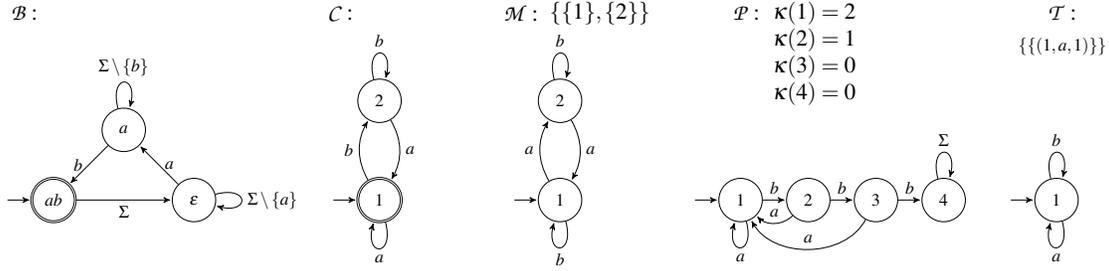 
%
%
We use three letter acronyms to describe acceptors, where the first letter is in $\{\textsc{d},\textsc{n}\}$ and denotes if the automaton is \emph{deterministic} or \emph{nondeterministic}. The second letter is one of $\{\textsc{b,c,p,m,t}\}$ for the first letter of the acceptance condition:  \buchi, \cobuchi, parity, Muller or \tmuller. The third letter is always $\textsc{a}$ for acceptor. Figure~\ref{fig:w-aut-examples} gives examples for a \dba, \dca, \dpa, \dma, and \dta, and specifies the accepted languages.
A language is said to be \emph{regular} if it is accepted by a \dfa. An $\omega$-language is said to be \emph{regular} if it is accepted by a \dma. An $\omega$-language is said to be \emph{weak} if it is accepted by a \dba\ as well as by a \dca.

\paragraph*{Complexity and expressiveness of sub-classes of regular $\omega$-languages}

We use $\DB$, $\DC$, $\DP$, $\DM$ and $\DT$ to denote the classes of languages accepted by \dba, \dca, \dpa, \dma\ and \dta, respectively, and $\NB$, $\NC$, $\NP$, $\NM$ and $\NT$ for the class of languages accepted by \nba, \nca, \npa, \nma\ and \nta, respectively. The classes $\NB$, $\NP$, $\NM$, $\NT$, $\DP$, $\DM$ and $\DT$ are equi-expressive and contain all $\omega$-regular languages. The classes $\DB$ and $\DC$ are strictly less expressive and are dual to each other in the sense that  $L\in\DB$ 
iff $L^c \in \DC$ where $L^c$ is the complement language of $L$, i.e. $\Sigma^\omega\setminus L$. The classes $\NC$ and $\DC$ are equi-expressive.

A subset $S$ of the automaton states, where for every $s_1,s_2\in S$ there exists a string $x \in \Sigma^+$ such that $\delta(s_1,x)=s_2$ is termed an SCC (abbreviating strongly connected component).\footnote{Note that there is no requirement for $S$ to be maximal in this sense.}
A Muller automaton $\M$ can be seen as classifying its SCCs into \emph{accepting} and \emph{rejecting}.
An important measure of the complexity of a Muller automaton is the number of alternations between accepting and rejecting SCCs along an inclusion chain. For instance, the Muller automaton $\M$ in Figure~\ref{fig:w-aut-examples} has an inclusion chain of SCCs with one alternation: $\{1\}$ (accepting) $\subseteq \{1,2\}$ (rejecting); and the Muller automaton $\T$ in Figure~\ref{fig:inIPnotinIM} whose only accepting SCC is $\{1,\lambda\}$ has an inclusion chain of SCCs with two alternations: $\{1\}$ (rejecting) $\subseteq \{1,\lambda\}$ (accepting) $\subseteq \{1,\lambda,0\}$ (rejecting). Wagner~\cite{Wagner75} has shown that this complexity measure is language-specific and is invariant over all Muller automata accepting the same language. Under this view, $\DB$ is the class of languages where no superset of an accepting SCC can be rejecting, $\DC$ is the class of languages where no subset of an accepting SCC can be rejecting, and $\DB \cap \DC$ is the class of languages where no alternation between accepting and rejecting SCCs is allowed along any inclusion chain. Thus, the language $\sema{\M}$ is not in $\DB$ and the language $\sema{\T}$ is not in $\DB \cup \DC$. 

\paragraph*{Right congruence and the rightcon automaton}
An equivalence relation $\sim$ on $\Sigma^*$ is a \emph{right congruence} if $x\sim y$ implies $xv \sim yv$ for every $x,y,v\in\Sigma^*$. The \emph{index} of $\sim$, denoted ${|\!\sim\!|}$ is the number of equivalence classes of $\sim$.  
Given a language $L$ its \emph{canonical right congruence} $\sim_L$ is defined as follows: 
$x \sim_L y$ iff ${\forall v \in \Sigma^*}$ we have ${xv\in L} \iff {yv \in L}$. 
For a word $v\in\Sigma^*$ 
the notation $[v]$ is used for the class of $\sim$ in which $v$ resides.

A right congruence $\sim$ can be naturally associated with an automaton $\M_\sim=\la \Sigma, Q, \lambda, \delta \ra$ as follows: the set of states $Q$ consists of the equivalence classes of $\sim$. The initial state $\lambda$ is the equivalence class $[\epsilon]$. The transition function $\delta$ is defined by $\delta([u],a)=[ua]$. In the sequel, we use $\R_L$ for the automaton $M_{\sim_L}$ associated with the right congruence of a given language $L$, and call it the \emph{rightcon automaton of $L$}.

Similarly, given an automaton $\M=\la \Sigma, Q, \lambda, \delta \ra$ we can naturally associate with it a right congruence as follows: $x \sim_\M y$ iff $\M(x)=\M(y)$. 
The Myhill-Nerode Theorem states that a language $L$ is regular  iff $\sim_L$ is of finite index. Moreover, if $L$ is accepted by a \dfa\ $\A$ then $\sim_\A$ refines $\sim_L$. Finally, the index of $\sim_L$ gives the size of the minimal \dfa\ for $L$. 

For $\omega$-languages, the right congruence $\sim_L$ is defined similarly, by quantifying over $\omega$-words. That is, $x \sim_L y$ iff ${\forall w \in \Sigma^\omega}$ we have ${xw\in L} \iff {yw \in L}$. Given a deterministic automaton $\M$ we can define $\sim_\M$ exactly as for finite words.
However, for $\omega$-regular languages, right congruence alone does not suffice to obtain a ``Myhill-Nerode'' characterization. As an example consider the language $L = (a+b)^*(aba)^\omega$. We have that $\sim_{L}$ consists of just one equivalence class, 
since for any $x\in\Sigma^*$ and $w\in\Sigma^\omega$ we have that $xw \in L$ iff $w$ has $(aba)^\omega$ as a suffix.
But an acceptor recognizing $L$ obviously needs more than a single state.
Note that the other side of the story entails that there are $\omega$-automata that are minimal, although two different states recognize the same language. For instance, the \dba\ $\B$ in Figure~\ref{fig:w-aut-examples} is minimal for $\sema{\B}$ but $\sema{\B_{\epsilon}}=\sema{\B_{a}}=\sema{\B_{{ab}}}$, and the \dma\ $\M$ of  Figure~\ref{fig:w-aut-examples} is minimal for $\sema{\M}$ but $\sema{\M_{1}}=\sema{\M_{2}}$. In general the problem of minimizing \dba s and \dpa s is known to be NP-complete~\cite{Schewe10}.

\paragraph*{Grammatical Inference}
\emph{Grammatical inference} or \emph{automata learning} refers to the problem of designing algorithms for inferring an unknown language from good and bad examples, i.e. from words labeled by their membership in the language. The learning algorithm is required to return some concise representation of the language, typically an automaton. The task of a learning algorithm can thus be thought of as trying to distinguish the different necessary states of an automaton recognizing the language and establishing the transitions between them. For a regular language, by the Myhill-Nerode theorem, $\sim_L$ can be used to distinguish states. Indeed, if the algorithm learns that $u_1v \in L$ and $u_2v \notin L$ for some $u_1,u_2,v\in\Sigma^*$ then $u_1\not\sim_L u_2$ and the words $u_1$ and $u_2$ must reach two different states of the minimal DFA for $L$. Once all equivalence classes of $\sim_L$ are discovered, the automaton $\M_{\sim_L}$ can be extracted, and by setting the state corresponding to the empty word to be the initial state, and states corresponding to positive examples as accepting, the minimal DFA is obtained. Many learning algorithms, e.g. Bierman and Feldman's algorithm~\cite{BF72} for learning a regular language from a finite sample, and Angluin's \lstar~\cite{Angluin87} algorithm for learning a regular language using membership and equivalence queries build on this idea. 

The idea of trying to distinguish states using right congruence relations is in the essence of many learning algorithms for formalisms richer than regular language (c.f.~\cite{AartsV10,Bojanczyk14,BB13,DrewsD17,MalerM17,MensM15}). For instance, learning of deterministic weighted automata~\cite{BM15,BV96} is founded on Fliess's theorem~\cite{Fli74} which is a generalization of the Myhill-Nerode theorem to the weighted automata setting.

For $\omega$-regular languages, learning algorithms encounter the problem that the right congruence is not informative enough. Maler and Pnueli~\cite{MalerPnueli95} give a polynomial time algorithm for learning the class $\DB\cap \DC$ using membership and equivalence queries. Their algorithm also works by trying to distinguish all equivalence classes of $\sim_L$ for the unknown language $L$, and relies on the fact that $\sim_L$ is informative enough for the class $\DB\cap \DC$. The problem of learning the full class of regular $\omega$-languages via membership and equivalence queries was considered open for many years~\cite{Leucker06full}. It was then suggested by Farzan et al.~\cite{FarzanCCTW08} to use the reduction to finite words~\cite{CalbrixNP93}, building on the fact that two $\omega$-languages are equivalent iff they agree on the set of ultimately periodic words. We followed a different route~\cite{AngluinF14,AngluinF16}, and devised an algorithm for learning the full class of $\omega$-regular languages using families of DFAs as acceptors~\cite{AngluinBF16,AngluinBF18}. Families of DFAs build on the notion of families of right congruences~\cite{MalerStaiger97}, a set of right congruence relations for a given regular $\omega$-language $L$, which are enough to fully characterize $L$.
Both solutions, however, may encounter big automata in intermediate stages. The solution in~\cite{FarzanCCTW08} may involve DFAs of size $2^n+2^{2n^2+n}$ where $n$ is the number of states in a non-deterministic \buchi\ automaton for the language~\cite{CalbrixNP93}, and the solution in~\cite{AngluinF14} may involve families of DFAs of size $m2^m$ where $m$ is the number of states in a minimal deterministic parity automaton for $L$~\cite{Fisman18}.
We do not go into further details here since they are not used in the current work, which focuses on languages for which the right congruence is informative enough (in hopes of providing a basis for more efficient learning algorithms for these restricted classes). The  reader interested in these details is referred to~\cite{Fisman18}. 

\paragraph*{Refinement of the right congruence}
\label{sec:relations-to-the-right-aut}

We show that as is the case in finitary regular languages, every deterministic $\omega$-automaton $\D$ refines the rightcon automaton of the respective language $\sema{\D}$, and the automaton for the powerset construction of a given non-deterministic $\omega$-automaton $\N$ refines the right congruence of the respective language $\sema{\N}$.

\begin{proposition}\label{prop:det-aut-refine-synt}
	Let $\D$ be a deterministic $\omega$-automaton. Then $\sim_{\D}$ refines $\sim_{\sema{\D}}$.
\end{proposition}

	Let $\N$ be a non-deterministic automaton. We use $\P_\N$ to denote the deterministic automaton obtained from $\N$ by applying the powerset construction to $\N$. That is, if $\N=(\Sigma,Q,q_0,\delta,\alpha)$ then $\P_\N=(\Sigma,2^Q,\{q_0\},\delta')$ where $\delta'(S,a)=\cup_{q\in S}\delta(q,a)$ for any $S\subseteq Q$ and $a\in \Sigma$. 

\begin{proposition}\label{prop:nondet-aut-refine-synt}
	Let $\N$ be a non-deterministic $\omega$-automaton. Then $\sim_{\P_\N}$ refines $\sim_{\sema{\N}}$.
\end{proposition}

\section{$\omega$-Languages with a trivial right congruence}
\label{sec:trivial}

If $L$ is a regular $\omega$-language such that ${|\sim_L|}=1$, we say that the rightcon automaton of $L$ is \emph{trivial}.
In this case, the rightcon automaton conveys almost no information about $L$.
It was shown in~\cite{Staiger83} that there are $2^{2^{\aleph_0}}$ $\omega$-languages for which the rightcon is trivial.
In Proposition~\ref{prop:characterization-for-trivial} we characterize those regular $\omega$-languages that have a trivial rightcon automaton.

If $w_1$ and $w_2$ are $\omega$-words, then $w_1$ is a \emph{finite variant} of $w_2$, denoted $w_1 =_\infty w_2$,  if there exist finite words $x_1$ and
$x_2$ and an $\omega$-word $w$ such that $w_1=x_1w$ and $w_2 = x_2w$.
The following shows that languages with a trivial rightcon automaton ignore differences between finite variants.
\begin{proposition}\label{prop:finite-variation}
	Let $L$ be a regular $\omega$-language such that ${|\sim_L|}=1$. Let $w_1,w_2\in\Sigma^\omega$.
	If $w_1 =_\infty w_2$ then $w_1 \in L$ iff $w_2 \in L$.
\end{proposition}

\begin{proof}
	Because $w_1=_\infty w_2$, there exist $x_1$, $x_2$ and $w$ such that $w_1 = x_1w$ and $w_2 = x_2w$.
	Because ${|\sim_L|}=1$, $[x_1]_{\sim_L} = [x_2]_{\sim_L}$. Thus $x_1w\in L$ iff $x_2w\in L$,
	so $w_1 \in L$ iff $w_2 \in L$.
\end{proof}

Clearly if $L=\Sigma^*v^\omega$ for some $v\in\Sigma^*$ then its rightcon automaton is trivial. Does this hold in general 
when $L=\Sigma^*L'$ for some $\omega$-regular language $L'$? The following example shows that the answer is negative.
Let  $\Sigma=\{a,b\}$ and $L_1=\Sigma^*L'$ for $L'=(ab^\omega + ba^\omega)$ then $\sim_{L_1}$ has four equivalence classes. 

However, if $L=\Sigma^*R^\omega$ for some regular language $R$ we can show that then $|\sim_L|=1$. Is this also a necessary condition for having a trivial right congruence? 
The answer again is negative. For example, the language $L_2 = (a+b)^*(a^{\omega} + b^{\omega}$) has
${|\sim_{L_2}|} = 1$, but it is not of the form $(a+b)^*R^{\omega}$ for
any regular set $R$.

The following proposition provides a full characterization of regular $\omega$-language with a trivial rightcon automaton.

\begin{proposition}\label{prop:characterization-for-trivial}
	A regular $\omega$-language has a trivial rightcon automaton iff $L=\Sigma^* (R_1^\omega + R_2^\omega + \ldots + R_k^\omega)$ for some regular languages $R_1,\ldots,R_k$.
\end{proposition}

\begin{proof}
	Let $\aut{M} =(\Sigma,Q,\lambda,\delta,\alpha)$ be a DMA for $L$ with no unreachable states. 
	Let $\alpha=\{S_1,S_2,\ldots,S_k\}$. We can assume wlog that all $S_i$'s are strongly connected (because an $S_i$ that is not an SCC can be omitted). For $i\in [1..k]$ let $s_i$ be some state in $S_i$ and let $R_i$ be the regular set of finite words that traverse $\aut{M}$ starting at $s_i$, ending in $s_i$, and visiting all states in $S_i$ and no other states. 
	
	Assume that ${|\sim_L|}=1$.
	Let $L' = \Sigma^*(R_1^\omega+R_2^\omega+\ldots+R_k^\omega)$. 
	We claim that $L = L'$. 
	Let $w \in L$. Then $w = xw'$, where for some $i$, $x \in \Sigma^*$ reaches $s_i \in S_i$ and $w' \in R_i^{\omega}$, so $w \in L'$.
	Conversely, if $w \in L'$ then $w = xw'$, where $x \in \Sigma^*$ and for some $i$, $w' \in R_i^{\omega}$.
	Because $M$ contains no unreachable states, there exists $y \in \Sigma^*$ such that $y$ reaches state $s_i \in S_i$.
	Then $yw' \in L$, so by Proposition~\ref{prop:finite-variation}, $xw' \in L$.
	
	For the converse, suppose that $R_1, \ldots, R_k$ are regular languages and $L = \Sigma^*(R_1^{\omega} + \ldots + R_k^{\omega})$.
	If ${|\sim_L|}>1$ then there exist $x, y \in \Sigma^*$ such that $[x]_L \neq [y]_L$, so wlog assume $w \in [x]_L$ and $w \not\in [y]_L$.
	Because $xw \in L$, there exists an $i$ such that $xw \in \Sigma^*R_i^{\omega}$.
	Thus for some $x' \in \Sigma^*$ and elements $w_1, w_2, \ldots$ of $R_i$, $xw = x'w_1w_2 \cdots$.
	Hence there exists $x'' \in \Sigma^*$ and $w'' \in R_i^{\omega}$ such that $xw = xx''w''$.
	But then $yw = yx''w''$, and $yx''w'' \in \Sigma^*R_i^{\omega}$,
	which implies that $yw \in L$, a contradiction.
	Thus we must have ${|\sim_L|}=1$.
\end{proof}

\section{$\omega$-Languages with an informative right congruence}
\label{sec:isomorphic}

We turn to examine the cases where the right congruence is as informative as it can be; that is the \rightaut\ automaton is isomorphic to an $\omega$-automaton recognizing the respective language.
We use $\IB$ (resp.  $\IC$, $\IP$, $\IM$, $\IT$) to denote the class of languages for which the rightcon automaton $\R_L$ is isomorphic to a \dba\ (resp.  \dca, \dpa, \dma, \dta) accepting the language $L$. 

\paragraph*{A small experiment}
We were curious to see what are the odds that a randomly generated Muller automaton will be isomorphic to its rightcon automaton, i.e. fully informative.
We ran a small experiment in which we generated a random Muller automaton over an alphabet of cardinality 3, with 2 accepting strongly connected sets, and tested whether it turned out to be isomorphic to its rightcon automaton.  
The procedure was to try to distinguish states of the random \dma\ using 100,000 random ultimately periodic $\omega$-words. If all states were successfully distinguished then the \dma\ is certainly isomorphic to its rightcon automaton, and was declared as such. If we failed to distinguish at least 2 states, we declared the \dma\ as non-isomorphic, though it might be that more tests would distinguish the undistinguished states and the \dma\ may in fact be isomorphic. So the probability of a randomly generated \dma\ being isomorphic to its rightcon automaton may be higher than what is suggested by our results.

We generated \dma s with 5, 6, 7, 8, 9, and 10 states; 100 of each size. The results are summarized in the following table. We find it interesting that in most of the cases a randomly generated \dma\ turns out to be isomorphic to its rightcon automaton, suggesting that this property is not rare. We defer a more careful study of the
extent to which random automata have informative
right congruences for further research. 

\begin{center}
	\begin{tabular}{|r||c|c|c|c|c|c|}
		\hline 
		Number of states & 5 & 6 & 7 & 8 & 9 & 10   \\ \hline \hline
		Isomorphic & 85 & 93 & 88 & 96 & 96 & 94 \\ 
		Not Isomorphic & 15 & 7 & 12 & 4 & 4 & 6\\ \hline
	\end{tabular} \\ 
\end{center}

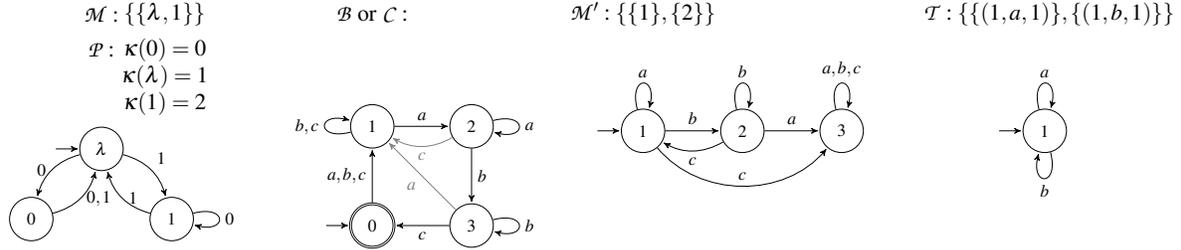
\begin{figure}[t]
	\begin{center}
		\scalebox{0.60}{
			\begin{tabular}{c@{\qquad\qquad\qquad}c@{\qquad\qquad\qquad}c}
				\begin{tikzpicture}[->,>=stealth',shorten >=1pt,auto,node distance=2.2cm,semithick,initial text=,initial where=left]
				
				\node[label]          (L)      [fontscale=1.5]          {$\aut{M}:$};
				\node[label]          (F)      [right of=L, node distance=1.4cm, fontscale=1.5]          {$\{\{\lambda,1\}\}$};
				\node[label]          (P)      [below of=L, node distance=.8cm, fontscale=1.5]          {$\aut{P}:$};
				\node[label]          (X1)      [right of=P, node distance=1.4cm, fontscale=1.5]          {$\kappa(0)=0$};
				\node[label]          (X2)      [below of=X1, node distance=.6cm, fontscale=1.5]          {$\kappa(\lambda)=1$};
				\node[label]          (X3)      [below of=X2, node distance=.6cm, fontscale=1.5]          {$\kappa(1)=2$};				
				
				\node[initial,state]  (Lambda)   [below of=L, node distance=3cm] {$\lambda$};
				\node[state]          (Zero)    [below left  of=Lambda]{$0$};
				\node[state]          (One)   [below right of=Lambda]{$1$};

				\path (Lambda) edge [bend right]  node [left] {$0$} (Zero); 
				\path (Lambda) edge [bend left]  node {$1$} (One); 
				\path (Zero) edge [bend right] node [right] {$0,1$} (Lambda); 
				\path (One) edge [loop right] node {$0$} (One); 
				\path (One) edge  [bend left]          node [right]{$1$} (Lambda);

				\node[label]          (A)   [right of=L, node distance=6cm, fontscale=1.5]              {$\B$ or $\C:$};
				
				\node[state]   (Le)  [below of=A, node distance=2.5cm]                          {$1$};
				\node[state]   (La)  [right of=Le]       {$2$};
				\node[state]   (Lb)  [below of=La]       {$3$};
				\node[state,accepting,initial]   (Lc)  [below of=Le]       {$0$};
				
				\path (Le) edge   node {$a$} (La);
				\path (La) edge   node {$b$} (Lb);
				\path (Lb) edge   node {$c$} (Lc);
				\path (Lc) edge   node {$a,b,c$} (Le);
				
				\path (Le) edge [loop left] node {$b,c$} (Le);
				\path (La) edge [loop right] node {$a$} (La);
				\path (Lb) edge [loop right] node {$b$} (Lb);
				
				\path (La) edge [bend left, gray] node {$c$} (Le);
				\path (Lb) edge [gray] node {$a$} (Le);

				\node[label] (xII) [right of=L, node distance=12cm, fontscale=1.5] {$\M': \{ \{1\}, \{2\}\}$};
				
				\node[initial,state]          (xQ0)   [below of=xII, node distance=2.6cm] {${1}$};
				\node[state] (xQ1)    [right  of=xQ0]{${2}$};
				\node[state] (xQ2)    [right  of=xQ1]{${3}$};
								
				\path (xQ0) edge [loop above] node {$a$} (xQ0); 
				\path (xQ0) edge           node {$b$} (xQ1); 
				\path (xQ0) edge  [bend right=50]         node {$c$} (xQ2); 
				
				\path (xQ1) edge [loop above] node {$b$} (xQ1); 
				\path (xQ1) edge  [bend left]          node {$c$} (xQ0); 
				\path (xQ1) edge           node {$a$} (xQ2); 
				\path (xQ2) edge  [loop above]         node {$a,b,c$} (xQ2);

				\node[label] (zII) [ right of=xII, node distance=9cm, fontscale=1.5] {$\T: \{ \{(1,a,1)\}, \{(1,b,1)\}\}$};
				
				\node[initial,state]          (zQ0)   [right of=xQ2, node distance=4.5cm] {${1}$};
								
				\path (zQ0) edge [loop above] node {$a$} (zQ0); 
				\path (zQ0) edge [loop below] node {$b$} (zQ0); 
				
				\end{tikzpicture}
				
		\end{tabular}}
		\caption{A \dma\ $\M$ and an equivalent \dpa\ $\P$ for a language $L_{PM}$ such that $L_{PM}\in\IM\setminus\DB\cap\DC$ and $L_{PM}\in\IP\setminus\DB\cap\DC$. Second from the left, when regarded as a \dba\ $\B$ we have $\sema{\B}\in\IB\setminus\DC$. When regarded as a \dca\ $\C$ we 
			have $\sema{\C}\in\IC\setminus\DB$. A \dma\ $\M'$ such that $\sema{\M'}\in\IM\setminus \IP$, and a \dta\ $\T$ such that $\sema{\T}\in\RT\setminus\RM$. }
		\label{fig:inDBWandIBnotinCBW}
		\label{fig:twisted-michel}
		\label{fig:inIPnotinIM}
	\end{center}
\end{figure}

\paragraph*{Expressiveness results}
As mentioned earlier, all weak regular $\omega$-languages, i.e. all languages that are in $\DB\cap\DC$ are isomorphic to their right congruence. 
We turn to the question of whether there exist languages outside this class that are isomorphic to their right congruence. 

Staiger~\cite{Staiger83} has shown that $\DB\cap\DC \subseteq \IM$. It is easy to see that this entails $\DB\cap\DC \subseteq \IP$. We show that both inclusions are strict.

\begin{proposition}
	\label{prop:DBcapDCsubsetneqIM}
	$\DB\cap\DC\subsetneq\IM$ and $\DB\cap\DC\subsetneq\IP$
\end{proposition}

\begin{proof}
	In~\cite[Thm. 24]{Staiger83} Staiger showed that $\DB\cap\DC \subseteq \IM$.  
	From this, since any \dba\ can be recognized by an isomorphic \dpa\ (by setting the accepting states color $1$ and the non-accepting states color $2$) and a \dca\ can be recognized by an isomorphic \dpa\ (by setting the $F$ states color $0$ and the non-$F$ states color $1$) it follows that $\DB\cap\DC \subseteq \IP$.
	
	To show that the inclusion is strict we show that the language $L_{PM}$ recognized by the automaton in Fig.~\ref{fig:twisted-michel} on the left, either when regarded as a \dma\ $\M$ or as a \dpa\ $\P$  is in $\IM \cap \IP$ but not in $\DB\cap\DC$. 
	One can verify that $\sim_{\Lang{\M}}$ is isomorphic to $\M$. (Note that $(0011)^{\omega}$ and $(0110)^{\omega}$ are sufficient
	to distinguish $\{\lambda, 0, 1\}$.)
	As mentioned in the preliminaries, this language is not in $\DB \cap \DC$ 
	since it has alternation between accepting and rejecting SCCs along
	the inclusion chain 
	$\{1\}\subseteq \{1,\lambda\} \subseteq \{1,\lambda,0\}$.
\end{proof}

It is easy to see that both $\IM$ and $\IP$ subsume both $\IB$ and $\IC$. The same example used in the proof of Proposition~\ref{prop:DBcapDCsubsetneqIM} can be used to show that the inclusion is strict.

\begin{proposition}
	\label{prop:IBandICsubsetneqIM}
	$\IB\cup\IC \subsetneq \IM$ and $\IB\cup\IC \subsetneq \IP$
\end{proposition}

It is thus interesting to see whether there are any languages in $\IB$ (or $\IC$) that are not already in $\DB\cap \DC$.
The answer is affirmative.
\begin{proposition}
	$\DB\cap\DC\subsetneq\IB$ and $\DB\cap\DC\subsetneq\IC$
\end{proposition}

\begin{proof}
	Assume $L \in \DB\cap\DC$.
	By Proposition~\ref{prop:DBcapDCsubsetneqIM} we have $L \in \IM$. Suppose $\M$ is a minimal  \dma\ for $L$. It follows, as explained in the preliminaries, that in $\M$ there is no alternation between accepting and rejecting SCCs along any inclusion chain.
	Therefore, defining a \dba\ $\B_\M$ from $\M$ by changing the acceptance condition to $\{ q ~|~ q \in S$ for some accepting SCC $S \}$ gives $\Lang{\B_\M}=\Lang{\M}$ and thus $L\in\IB$.  Similarly,  defining a \dca\ $\C_\M$ from $\M$ by changing the acceptance condition to $\{ q ~|~ q \in S$ for some rejecting SCC $S\}$ gives $\Lang{\C_\M}=\Lang{\M}$ and thus $L\in\IC$. This completes the inclusion part.
	
	To see that the inclusion is strict for $\IB$ consider the \dba\ $\B$ from Fig.~\ref{fig:twisted-michel}. By~\cite[Lemma 2]{KupfermanMM06} for every language $L$ recognized by a \dba\ $\B$, if $L$ is also in $\DC$, then a \dca\ embodied in the structure of $\B$ can be defined. Since none of the \dca s embodied in $\B$ accepts the same language as $\Lang{\B}$, it follows that $\Lang{\B} \in \DB\setminus \DC$. Thus $\B\notin\DB\cap\DC$.
	
	Next we show that $\sim_{\Lang{\B}}$ has at least 4 equivalence classes. The $\omega$-word $(ababc)^\omega$ is accepted only from states $0$ and $3$ and the $\omega$-word $(babca)^\omega$ is accepted only from states $0$ and $1$; these two experiments distinguish all 4 states.
	Since by Proposition~\ref{prop:det-aut-refine-synt} $\B$ refines $\sim_{\Lang{\B}}$, it follows that the automaton for $\sim_{\Lang{\B}}$ is isomorphic to $\B$, thus $\Lang{\B}\in\IB$. 
	
	The proof for strictness for $\IC$ is dual, using the \dca\ $\C$ in Fig.~\ref{fig:twisted-michel} and~\cite[Lemma 2]{KupfermanMM06} which states also that for every language $L$ recognized by a \dca\ $\C$ if $L$ is also in $\DB$ then a \dba\ embodied in the structure of $\C$ can be defined.
\end{proof}

Since $\IB \subseteq \DB$ and $\IC \subseteq \DC$, from $\DB \cap \DC \subseteq \IB \cap \IC$ we get that $\DB \cap \DC = \IB \cap \IC$.

\begin{corollary}
	$\DB \cap \DC = \IB \cap \IC$
\end{corollary}

It is shown in~\cite{VanSL95} that any \dma\ can be defined on a \dta\ with the same structure, and that the converse is not true. For instance, the language $(a+b)^*a^\omega$ can be defined by the one-state \dta\ $\T$ in Fig.~\ref{fig:dta}, but no \dma\ with one state accepts it. This shows $\IM$ is strictly contained in $\IT$.

\begin{proposition}[\cite{VanSL95}]\label{prop:ITIM}
	$\IM \subsetneq \IT$
\end{proposition}

The last missing part of the puzzle of the inclusions of the subsets \IB, \IC, \IP, \IM, \IT\ is provided in the following proposition  showing that $\IP$ is strictly contained in $\IM$.

\begin{proposition}\label{prop:IMIP}
	$\IP\subsetneq\IM$ 
\end{proposition}

\begin{proof}
	Inclusion follows since any \dpa\ can be converted into a \dma\ on the same structure (by setting an SCC to accepting iff the minimal color in it is odd). We claim that the \dma\ $\M'$ in Figure~\ref{fig:inIPnotinIM} is in $\IM \setminus \IP$. To see that it is in $\IM$, note that $ca^\omega$ distinguishes state $2$ from the other states, and $a^\omega$ distinguishes states $1$ and $3$. Thus the rightcon automaton has $3$ states, and since by Proposition~\ref{prop:det-aut-refine-synt}, $\M'$ refines $\R_{\sema{\M'}}$ we get that they are isomorphic. To see that it is not in $\IP$, note that to define a \dpa\ on the same structure we need to give state $1$ an odd color, so that when the set of states visited inf. often is $\{1\}$ it will accept. For the same reason we need to give state $2$ an odd color. But then, when the set of states visited infinitely often is $\{1,2\}$, the automaton will accept as well, while it needs to reject. 
\end{proof}

These relations are summarized in Figure~\ref{fig:ibicimip-relations} on the left. The last question is then how complex can a language in \IP, \IM, or \IT\ be? We show that such languages can be arbitrarily complex. That is, for every class $\DM_{n,m}^p$ of the Wagner Hierarchy, there is a language $L\in \IM\cap\IP\cap\IT$ that is in $\DM_{n,m}^p$ and not in any proper subclass of the Wagner Hierarchy. Thus, $\IT$, $\IM$ and $\IP$ include classes as complex as can be\footnote{A similar result is mentioned without a proof in  a footnote in~\cite{MalerStaiger97}, with credit to ``N. Gutleben (personal communication)''. }, as measured by the Wagner Hierarchy.\footnote{Due to lack of space we do not include the full definition of the Wagner Hierarchy. For details, we refer the reader to~\cite{Wagner75},\cite[Chapter V]{PerrinPin04}. }

\begin{proposition}
	\label{prop:im-ip-wagner}
	Let $n,m$ be two natural numbers and $p\in\{+,-,\pm\}$.
	Let $\DM_{n,m}^p$ denote the Wagner Hierarchy class with a maximum of $n$ alternations in each inclusion chain starting with polariy $p$, and a sequence of at most $m$ chains of alternating polarity. Then exists a language $L\in\Sigma^\omega$ for $\Sigma =\{a,b\}$ such that 
	$L\in \IT\cap \IM \cap \IP \cap \DM_{n,m}^p$ and $L\notin  \DM_{n\!-\!1,m}^p$ and $L\notin  \DM_{n,m\!-\!1}^p$.
\end{proposition}

\begin{figure}[t]
	\noindent\makebox[\textwidth]{
		\begin{tabular}{lr}
			\hspace{-3mm}
			\scalebox{0.54}{
				\begin{tikzpicture}[->,>=stealth',shorten >=1pt,auto,node distance=2.2cm,semithick,initial text=]

\node[label]            (Aut)               {$\mathlarger{\mathlarger{\mathlarger{\D_{n,m}^{+}}:}}$};
\node[initial, state] (A)  [below of=Aut] {$0^0$};
\node[state]          (B)  [right of=A]   {$1^0$};
\node[state]          (C)  [right of=B]   {$2^0$};
\node[label]          (CP)  [right of=C]   {$\circ\circ\circ$};
\node[state]          (CL)  [right of=CP]   {$n^0$};

\node[state]          (D)  [below of=A, node distance=2.5cm]   {$0^1$};
\node[state]          (E)  [right of=D]   {$1^1$};
\node[state]          (F)  [right of=E]   {$2^1$};
\node[label]          (FP)  [right of=F]   {$\circ\circ\circ$};
\node[state]          (G)  [right of=FP]   {$n^1$};

\node[label]          (H0)  [below of=D, node distance=2.5cm]   {};
\node[label]          (H)  [below of=F, node distance=2.5cm]   {$\circ\circ\circ$};
\node[label]          (H1)  [below of=G, node distance=2.5cm]   {};

\node[state]          (I)  [below of=D, node distance=5cm] {$0^m$};
\node[state]          (J)  [right of=I]   {$1^m$};
\node[state]          (K)  [right of=J]   {$2^m$};
\node[label]          (L)  [right of=K]   {$\circ\circ\circ$};
\node[state]          (M)  [right of=L]   {$n^m$};

\path (A) edge [loop below]          node {$a$} (A)
          edge                       node {$b$} (B)
      (B) edge [bend left=45,above]  node {$a$} (A)
          edge                       node {$b$} (C)
      (C) edge [bend left=62,above]  node {$a$} (A)
          edge                       node {$b$} (CP)
      (D) edge                       node {$b$} (E)
          edge [loop below]          node {$a$} (D)
      (E) edge [bend left=45,above]  node {$a$} (D)
          edge                       node {$b$} (F)
      (F) edge [bend left=62,above]  node {$a$} (D)
          edge                       node {$b$} (FP)
      (M) edge [loop right]          node {$\Sigma$} (M)
;

\path (CL) edge  node [above]{$b$} (D);
\path (G) edge [dotted] node [above]{$b$} (H0);
\path (H1) edge [dotted] node [above]{$b$} (I);

\path (CP) edge node {$b$} (CL);
\path (FP) edge node {$b$} (G);
\path (L) edge node {$b$} (M);

\path (I) edge                       node {$b$} (J)
          edge [loop below]          node {$a$} (I)
      (J) edge [bend left=45,above]  node {$a$} (I)
          edge                       node {$b$} (K)
      (K) edge [bend left=62,above]  node {$a$} (I)
          edge                       node {$b$} (L)
;      
     
\end{tikzpicture}
			}
			
			&
			\hspace{-4mm}
			\scalebox{0.54}{
			\begin{tikzpicture}[->,>=stealth',shorten >=1pt,auto,node distance=2.2cm,semithick,initial text=, initial where=above]
			
			\node[label] (B)   [fontscale=1.5]      {$\B_{bad}:$};
			
			\node[initial, state, accepting]  (Blambda)   [below right of=B] {${\lambda}$};
			\node[state] (B0)    [left  of=Blambda]{${0}$};
			\node[state] (B1)    [above right of=Blambda]{${1}$};
			\node[state] (B2)    [below right of=Blambda]{${2}$};
			
			\path (Blambda) edge  node  {$0$} (B0);
			\path (Blambda) edge [<->] node [above]{$1$} (B1); 
			\path (Blambda) edge [<->]  node [below] {$2$} (B2); 
			\path (B2) edge [loop right]  node {$0,1$} (B2);  
			\path (B1) edge [loop right]  node {$0,2$} (B2); 
			\path (B0) edge [loop left]  node {$0,1,2$} (B0);

			\node[label] (C)   [right of=B, node distance=7cm, fontscale=1.5]      {$\C_{bad}:$};

			\node[initial, state]  (Clambda)   [below right of=C] {${\lambda}$};
			\node[state, accepting] (C0)    [left  of=Clambda]{${0}$};
			\node[state] (C1)    [above right of=Clambda]{${1}$};
			\node[state] (C2)    [below right of=Clambda]{${2}$};
			\node[state, accepting] (C3)    [below right of=C1]{${3}$};

			\path (Clambda) edge  node  {$0$} (C0);
			\path (Clambda) edge [<->] node [above]{$1$} (C1); 
			\path (Clambda) edge [<->]  node [below] {$2$} (C2); 
			\path (C2) edge [loop right]  node {$0,1,3$} (C2); 
			\path (C1) edge [loop right]  node {$0,2$} (C2); 
			\path (C0) edge [loop left]  node {$0,1,2,3$} (C0); 
			\path (Clambda) edge [loop below]  node {$3$} (Clambda); 
			\path (C3) edge [loop right]  node {$0,1,2$} (C3);
			\path (C1) edge [<->] node [above]{$3$} (C3); 

			\node[label] (Di)   [below of=B, node distance=5.5cm]     {};	
			\node[label] (D)   [right of=Di, node distance=2cm, fontscale=1.5]     {$\D_{bad}:$};
			
			\node[initial, state, accepting]  (Dlambda)   [below right of=D] {${\lambda}$};
			\node[state] (D0)    [left  of=Dlambda]{${0}$};
			\node[state, accepting] (D1)    [above right of=Dlambda]{${1}$};
			\node[state, accepting] (D2)    [below right of=Dlambda]{${2}$};
			\node[state] (D3)    [right of=D1]{${3}$};
			\node[state, accepting] (D4)    [below right of=D1]{${4}$};

			\path (Dlambda) edge  node  {$0$} (D0);
			\path (Dlambda) edge [<->] node [above]{$1$} (D1); 
			\path (Dlambda) edge [<->]  node [below] {$2$} (D2); 
			\path (D2) edge [loop right]  node {$0,1,3$} (D2); 
			\path (D1) edge [loop above]  node {$0,2,4$} (D2); 
			\path (D0) edge [loop left]  node {$0,1,2,3,4$} (D0); 
			\path (Dlambda) edge [loop below]  node {$3,4$} (Dlambda); 
			\path (D3) edge [loop right]  node {$0,1,2,3$} (D3);
			\path (D1) edge [->] node [above]{$3$} (D3); 			 
			\path (D1) edge [->] node [above]{$4$} (D4); 			 
			\path (D3) edge [->] node {$4$} (D4); 			 
			\path (D4) edge [loop right]  node {$0,1,2,3,4$} (D4);
			\path (D2) edge [bend left] node {$4$} (D0); 			  				
						
			\end{tikzpicture}}
	\end{tabular}}
	\caption{On the left, a representative example for the Wagner class $\DM_{n,m}^+$. The acceptance condition is $\F = \{ \{ 0^i, 1^i, \ldots, j^i\}~|~ j \mbox{ is odd iff } i \mbox{ is odd} \}$. 
		On the right a \dba\ $\B_{bad}$ and \dca\ $\C_{bad}$ in $\IB \setminus \RB$ and $\IC\setminus \RC$, respectively, and a \dba\ $\D_{bad}$ in $\IB \cap \IC$ that is not respective of its right congruence. }\label{fig:wag-example-and-bads}
\end{figure}
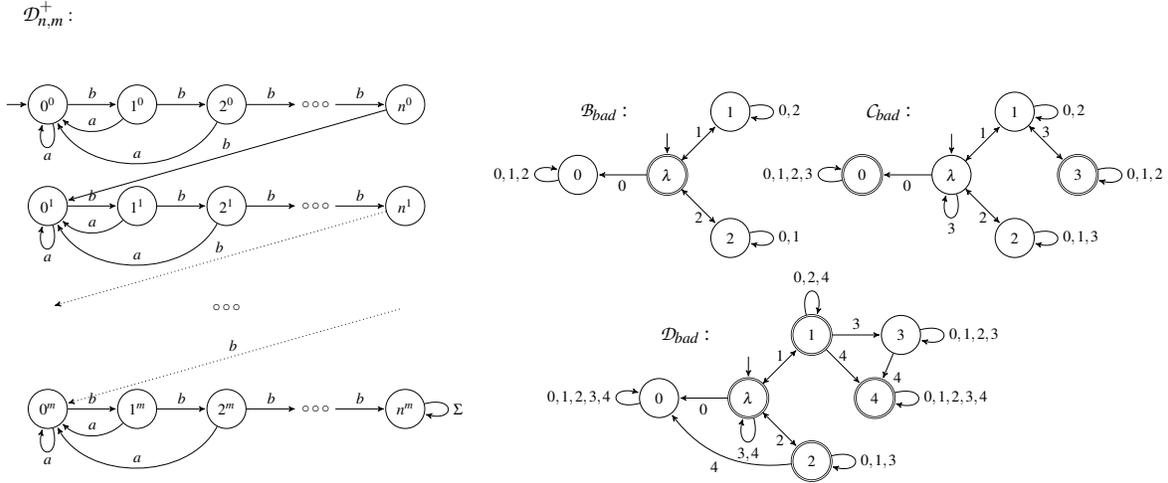

\begin{proof}
	Consider the \dma\ $\D_{n,m}^{+}$ in Figure~\ref{fig:wag-example-and-bads}, with acceptance condition $\F = \{ \{ 0^i, 1^i, \ldots, j^i\}~|~j$  is odd iff $i$ is odd $\}$. For instance, $\{0^0\}\in \F$, $\{0^0,1^0,2^0\}\in \F$, and $\{0^3,1^3\} \in \F$ but $\{1^0\}\notin \F$ and $\{0^0,1^0\}\notin\F$. It strictly belongs to the Wagner hierarchy class $\DM_{n,m}^{+}$. To see that it is in $\IM$, we show for each state, a word that distinguishes it from other states. We fix an order between the states: a state $k^\ell$ is smaller than $k'^{\ell'}$ if either $\ell<\ell'$ or $\ell=\ell'$ and $k<k'$. Thus the last state $n^m$ is the biggest in this order.
	The word $a^\omega$ distinguishes the last state $n^m$ from all states on odd rows, and the word $(ab)^\omega$ distinguishes $m^n$ from all states on even rows. For $k\in[0..n]$, the word $b^{n-k}a^\omega$ distinguishes state $k^m$ from all smaller states on odd rows, and the word $b^{n-k}(ab)^\omega$ distinguishes it from all smaller states on even rows. Finally, for $k\in[0..n]$ and $\ell\in[0..m]$ the word $(b^{n+1})^{m-\ell} b^{n-k}a^\omega$ distinguishes state $k^\ell$ from all smaller states on odd rows, and the word $(b^{n+1})^{m-\ell} b^{n-k}(ab)^\omega$ distinguishes it from all smaller states on even rows. This shows that the rightcon automaton has $(n+1)(m+1)$ states, and thus $\sema{\D_{n,m}^+}\in\IM$. It is also in $\IP$, since we can define a \dpa\ on the same structure, by assigning state $k^\ell$ the color $k$ if $\ell$ is odd, and $k+1$ if $\ell$ is even. It is in $\IT$ since $\IM \subset \IT$.
	
	The proof for  $\DM_{n,m}^{-}$ is symmetric, and the proof for $\DM_{n,m}^{\pm}$ can be easily deduced from this.
\end{proof}

\section{Respective of the right congruence}
\label{sec:respective}
As mentioned above, one of the motivations for studying classes of languages that are isomorphic to the right congruence is in the context of learning an unknown language. In this context, positive and negative examples ($\omega$-words labeled by their membership in the language) should help a learning algorithm to infer an automaton for the language. Consider the positive example $(ab)^\omega$ for an unknown language $L$. Intuitively, we expect that a minimal automaton for $L$ would have a loop of size $2$ in which the word $ab$ cycles. This is not necessarily the case, as shown by the language $L=(aba+bab)^\omega$, whose minimal \dba\ $\B_{\bowtie}$ is given in Fig~\ref{fig:bowtie}. In the case of regular languages of finite words, if we regard the automaton $\B_{\bowtie}$  as a \dfa, we note that $ab$, $abab$ are negative examples, while $ababab$ is a positive example. From this a learning algorithm can clearly infer the smallest loop on which $ab$ cycles is of length $6$, and not $2$. But in the case of $\omega$-languages there are no negative $\omega$-words that can provide such information. We thus define a class of languages in which if $uv^\omega$ is a positive example for $L$, then a minimal automaton for $L$ has a cycle of length at most $|v|$ in which $uv^\omega$ loops.

\begin{definition}[respective of $\sim_L$]
	
	A language $L$ is said to be \emph{respective} of its right congruence if $\exists n_0\in \naturals.$ $\forall n>n_0.$ $\forall x,u\in\Sigma^*.$  $xu^\omega \in L$ implies $xu^{n} \sim_L xu^{n+1}$. 
\end{definition}

Intuitively, a language that is respective of its right congruence, can ``delay'' entering a loop as much as needed, but once it loops on a periodic part, it loops on the smallest period possible.

Being respective of the right congruence does not entail having an $\omega$-automaton that is isomorphic to the right congruence. 
Any language $L$ with ${|\sim_L|}=1$ is (trivially) respective of its right congruence.
By Proposition~\ref{prop:characterization-for-trivial}, $L = (a+b)^*(aba)^{\omega}$ has ${|\sim_L|}=1$,
but $L$ is not in $\IT$, $\IM$, $\IP$, $\IB$ or $\IC$, because every $\omega$-automaton accepting $L$ requires more than one state.
We thus concentrate on languages which are both isomorphic to the rightcon automaton and respective of their right congruence.
We use $\RB$, $\RC$, $\RP$, $\RM$ and $\RT$ for the classes of languages that are respective of their right congruence and reside in $\IB$, $\IC$, $\IP$, $\IM$ and $\IT$, respectively.
By definition, thus, $\class{IX} \supseteq \class{RX}$ for $\class{X} \in \{ \class{B},\class{C},\class{P},\class{M},\class{T}\}$. We show that these inclusions are strict. 

\begin{proposition}\label{prop:IXstrictlyRX}
	$\IB \supsetneq \RB$, 	$\IC \supsetneq \RC$,  $\IP \supsetneq \RP$, $\IM \supsetneq \RM$ and $\IT \supsetneq \RT$.
\end{proposition}

\begin{proof}
	Consider the \dba\ $\B_{bad}$ in Fig.~\ref{fig:wag-example-and-bads} on the right. The language  $B_{bad}$ accepted by $\B_{bad}$ is in $\IB\subset\IP \subset \IM \subset \IT$ but it is not respective of its right congruence. To see that it is in $\IB$ take $\epsilon$, $0$, $1$ and $2$ as the representative words for states $\lambda$, $0$, $1$ and $2$ respectively. Note that $(011)^\omega$ distinguishes $1$ from the rest of the representative words, and $(022)^\omega$ distinguishes $2$ from the rest of the representative words. Finally, $(11)^\omega$ distinguishes $\epsilon$ from $0$. The pair $(\epsilon,1012)$ shows that $B_{bad}$ is not respective of its right congruence since $(1012)^\omega \in B_{bad}$ yet for all $n\in\naturals$ we have that $(1012)^{n+1} \not\sim_{B_{bad}} (1012)^n$. 
	
	Consider the \dca\ $\C_{bad}$  in Fig.~\ref{fig:wag-example-and-bads}. The language  $C_{bad}$ accepted by $\C_{bad}$ is in $\IC$ but it is not respective of its right congruence. To see that it is in $\IC$ take $\epsilon$, $0$, $1$, $2$ and $13$ as the representative words for states $\lambda$, $0$, $1$, $2$ and $3$ respectively. Note that $0^\omega$ distinguishes $1$ and $2$ from the rest of the representative words; $3^\omega$ distinguishes $1$ from $2$ and it distinguishes $\epsilon$ from $0$ and $13$; and $30^\omega$ distinguishes $13$ from $0$. The pair $(\epsilon,1012)$ shows that $C_{bad}$ is not respective of its right congruence since $(1012)^\omega \in C_{bad}$ yet for all $n\in\naturals$ we have that $(1012)^{n+1} \not\sim_{C_{bad}} (1012)^n$. 
\end{proof}

Recall that $\DB\cap \DC = \IB \cap \IC$. The \dba\ $\D_{bad}$ in Fig.~\ref{fig:wag-example-and-bads} can be used to show a language in $\IB \cap \IC$  that is not respective of its right congruence.

\begin{proposition}	
	\label{prop:IBandICnotRespective}
	There exists languages in $\IB \cap \IC$ that are not respective of their right congruences. 
\end{proposition}

To complete the picture of inclusions between the $\class{RX}$ classes, we establish that $\RM \supsetneq \RP$ and $\RT \supsetneq \RM$.

\begin{proposition}	\label{prop:RMminusRP}
	$\RM \supsetneq \RP$ and $\RT \supsetneq \RM$. 
\end{proposition}

Figure~\ref{fig:RXinclusions} on the right summarizes the above results. While it is notable that the requirement of being respective of the right congruence constitutes a restriction,  there exist languages which are respective of their right congruence in every class of the Wagner Hierarchy.

\begin{proposition}
	\label{prop:rm-rp-wagner}
	Let $n,m$ be two natural numbers and $p\in\{+,-,\pm\}$.
	Let $\DM_{n,m}^p$ denote the Wagner Hierarchy class with a maximum of $n$ alternations in each inclusion chain, and a sequence of at most $m$ chains of alternating polarity. Then there exists a language $L\in\Sigma^\omega$ for $\Sigma =\{a,b\}$ such that 
	$L\in \DM_{n,m}^p \cap \RM \cap \RP \cap \RT$.
\end{proposition}

\begin{proof}
	Consider again the \dma\ $\D_{n,m}^+$ in Fig.~\ref{fig:wag-example-and-bads} and let ${L_{n,m}^+}$ be the language that it recognizes. We have established in the proof of Proposition~\ref{prop:im-ip-wagner} that $\sema{\D_{n,m}^+} \in \IM \cap \IP \cap \IT$. Take $n_0=(m+1)(n+1)$. 
	Consider the state that  $\D_{n,m}^+$ reaches after reading $xu^{n_0}$.
	If this state is $n^m$ then clearly for any $n'>n_0$, $xu^{n'}$ also reaches states $n^m$. Otherwise there is an $a$ following the longest subsequence of $b^*$ in $u$. The state that $\D_{n,m}^+$ reaches after reading $xu^{n_0}$ depends on (a) the the longest subsequence of $b$'s in $x$ (b) the longest subsequence of $b$'s in $u$ and (c) the number of consecutive $b$'s in the rightmost subsequence of $b$'s in $u$. Since the parameter (a) depends only on $x$ and the parameters (b) and (c) remain the same in $u^i$ for any $i\in\naturals$ we have that $xu^{n_0+i}\sim_{L_{n,m}^+} xu^{n_0+i+1}$.
	Thus, the accepted language is respective of its right congruence.
\end{proof}

\paragraph*{Relation to Non-Counting Languages}
The definition of respective of $\sim_L$ is reminiscent of the definition of non-counting languages~\cite{DiekertG08}. 
A language $L\subseteq \Sigma^\omega$ is said to be \emph{non-counting} iff $\exists n_0 \in\naturals.\ \forall n > n_0.\ \forall u,v\in\Sigma^*,\ w\in\Sigma^\omega. uv^nw\in L \iff uv^{n+1}w \in L$. 

\begin{proposition} 
	\label{prop:non-counting-implies-respective}
	If $L$ is non-counting then $L$ is respective of its right-congruence.
\end{proposition}
The converse does not hold. The set $(aa)^*b^\omega$ is respective of $\sim_L$ but is not non-counting.

\begin{proposition}
	\label{prop:respective-not-imply-non-counting}
	There exist languages that are respective of $\sim_L$ but are not non-counting. 
\end{proposition}

One of the most commonly used temporal logics is Linear temporal logic (LTL)~\cite{Pnueli77}.
LTL formulas are non-counting~\cite{GabbayPSS80,Wolper83,DiekertG08,Rabinovich14}. But, there are LTL formulas that characterize languages that are not in $\IT$ (and thus not in any of the $\class{I}$ classes). Indeed, the formula $FG\,(a \vee X\,a)$ characterizes the language $L=\Sigma^*(a+\Sigma a)^\omega$ and by Proposition~\ref{prop:characterization-for-trivial}, ${|\sim_L|=1}$.

\section{Closure properties}
\label{subsec:closure}
\vspace{-1mm}
We examine what Boolean closure properties hold or do not hold for these classes of languages recognizable by an automaton isomorphic to the \rightaut\ automaton, and by automata that are also respective of the right congruence.

It is a well known result that weak regular $\omega$-languages are closed under all Boolean operations as stated by the following proposition.

\begin{proposition}[c.f.~\cite{MannaP89}]
	\label{prop:dbcapdc-closed-to-all}
	The class $\DB\cap\DC$ is closed under the Boolean operations complementation, union and intersection.
\end{proposition}

The classes $\IT$, $\IM$ and $\IP$ are closed under complementation. The other classes that we consider are not closed under complementation. See Fig.~\ref{fig:r-classes-comp-non-closure} for counterexamples.
\begin{proposition}
	\label{prop:im-ip-comp}
	The classes $\IT$, $\IM$ and $\IP$ are closed under complementation.\\
	The classes $\IB$ and $\IC$ are not closed under complementation.
\end{proposition}

\begin{proposition}
	\label{prop:IB-IC-comp}
	$\RB$, $\RC$, $\RP$, $\RM$ and $\RT$ are not closed under complementation.
\end{proposition}

\begin{proof}
	Consider the \dba\ $\P$ in Fig.~\ref{fig:r-classes-comp-non-closure}. The $\omega$-words $ba(ac)^\omega$, $a(ac)^\omega$, $(ac)^\omega$ and $(ca)^\omega$ distinguish its five states (the four shown in the figure, and the sink state). This shows that $\P\in \IB \subset \IP \subset \IM \subset \IT$. 
	To see that it is respective of its right congruence,
	if $xy^\omega$ is in $\sema{\P}$ then $y$ must traverse the cycle of states 3 and 4.
	Thus for some $n_0$, $\P(xy^{n_0}) = 3$ and $y = (ac)^k$ for some $k \ge 1$, or $\P(xy^{n_0}) = 4$  and $y = (ca)^k$ for some $k \ge 1$.
	In either case, $\P(xy^n) = \P(xy^{n+1})$ for all $n \ge n_0$.
	However, its complement $\P^c$ accepts the word 
	$b^\omega$. But for every $n\in\naturals $ we have $b^n \not\sim_{\sema{\P^c}} b^{n+1}$.
	
	This shows that $\RB$, $\RP$, $\RM$ and $\RT$ are not closed under complementation. Consider the $\dca$ $\P'$ obtained from $\P$ by marking the states $\{1,2,0\}$ where $0$ is the sink state. Then $\P'$ recognizes the same language as $\P$. We get that $\P'$ is in $\IC$, is respective of its right congruence, yet its complement is not respective of its right congruence.
\end{proof}

Aside from the  class $\DB\cap \DC$ none of the classes are closed under union or intersection. The automata in Figures~\ref{fig:IB-IC-non-closure} and~\ref{fig:union-non-closure} are used to refute these closures. The complete proofs are given in the full version.

\begin{proposition}
	\label{prop:all-union}
	\label{prop:all-intersection}
	The classes $\IB$, $\IC$, $\IP$, $\IM$  and $\IT$ are not closed under union or intersection.
	The classes $\RB$, $\RC$, $\RP$, $\RM$ and $\RT$ are not closed under union or intersection.
\end{proposition}


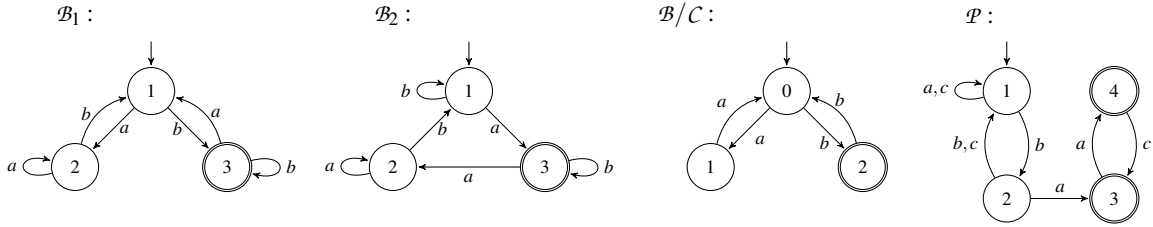
\begin{figure}[t]
	\noindent\makebox[\textwidth]{
		\scalebox{0.65}{
			\begin{tikzpicture}[->,>=stealth',shorten >=1pt,auto,node distance=2.2cm,semithick,initial text=, initial where=above]
			
			\node[label] (B)     [fontscale=1.5]    {$\B_1:$};
			
			\node[initial,state]  (B1)   [below right of=B] {${1}$};
			\node[state] (B2)    [below left  of=B1]{${2}$};
			\node[state,accepting] (B3)    [below right of=B1]{${3}$};
			
			\path (B1) edge  node [right] {$a$} (B2);
			\path (B1) edge  node [left] {$b$} (B3); 
			\path (B2) edge [loop left]  node {$a$} (B2); 
			\path (B2) edge [bend left]  node [left] {$b$} (B1); 
			\path (B3) edge [bend right]  node [right] {$a$} (B1); 
			\path (B3) edge [loop right]  node [right] {$b$} (B3);

			\node[label] (sB)     [right of=B, node distance=6.5cm, fontscale=1.5]    {$\B_2:$};
			
			\node[initial,state]  (sB1)   [below right of=sB] {${1}$};
			\node[state] (sB2)    [below left of=sB1]{${2}$};
			\node[state,accepting] (sB3)    [below right of=sB1]{${3}$};
			
			\path (sB1) edge [loop left]  node {$b$} (sB1); 
			\path (sB1) edge  node [left] {$a$} (sB3); 
			\path (sB3) edge [loop right]  node {$b$} (sB3); 
			\path (sB3) edge  node [below] {$a$} (sB2); 
			\path (sB2) edge [loop left]  node {$a$} (sB2); 
			\path (sB2) edge  node [right] {$b$} (sB1); 
			
			\node[label]          (xLL)   [right of=B, node distance=12.5cm, fontscale=1.5]              {$\B/\C:$};
			\node[initial,state]  (xQ0)   [right of=sB1, node distance=6.5cm] {$0$};
			\node[state]  (xQ1)   [below left of=xQ0] {$1$};
			\node[state,accepting]  (xQ2)   [below right of=xQ0] {$2$};
			
			\path (xQ0) edge  node {$a$} (xQ1); 
			\path (xQ0) edge   node [below] {$b$} (xQ2); 
			\path (xQ1) edge [bend left] node {$a$} (xQ0); 
			\path (xQ2) edge [bend right] node [above]{$b$} (xQ0); 
			
			\node[label]          (yLL)   [right of=B, node distance=18.5cm, fontscale=1.5]  {$\P:$};
			\node[initial,state]  (yQ1)   [right of=sB1, node distance=11cm] {$1$};
			\node[state]  (yQ2)   [ below of=yQ1] {$2$};
			\node[state,accepting]  (yQ3)   [ right of=yQ2] {$3$};
			\node[state,accepting]  (yQ4)   [ above of=yQ3] {$4$};

			\path (yQ1) edge [loop left] node {$a,c$} (yQ1); 
			\path (yQ1) edge [bend left] node {$b$} (yQ2); 
			\path (yQ2) edge [bend left] node {$b,c$} (yQ1); 
			\path (yQ2) edge node {$a$} (yQ3); 
			\path (yQ3) edge  [bend left]  node {$a$} (yQ4); 
			\path (yQ4) edge [bend left] node {$c$} (yQ3); 

			\end{tikzpicture}}
	}
	\caption{On the left examples for non-closure of union for \IB. On the right a \dba\ $\B$ and a \dca\ $\C$ showing $\IB$ and $\IC$ are not closed under complementation (the sink state is not shown).  }
	\label{fig:IB-IC-non-closure}
	\label{fig:ib-union-non-closure}
	\label{fig:r-classes-comp-non-closure}
\end{figure}

\begin{figure}[t]
	\noindent\makebox[\textwidth]{
		\scalebox{0.6}{
			\begin{tikzpicture}[->,>=stealth',shorten >=1pt,auto,node distance=2.2cm,semithick,initial text=]
			
			\node[label] (M)    [fontscale=1.5]     {$\M_1:\{ \{0\}\}$};
			\node[label] (P) [below of=M, node distance=.7cm, fontscale=1.5]  {$\P_1: \kappa(0)=1, \kappa(1)=2$};
			\node[label] (C) [below of=P, node distance=.7cm, fontscale=1.5]  {$\C_1: \{ 1 \}$};
			
			\node[initial,state]  (xQ0)   [below left of=C, node distance=2.2cm] {${0}$};
			\node[state] (xQ1)    [right  of=xQ0]{${1}$};
			
			\path (xQ0) edge [loop below] node {$b,c$} (xQ0); 
			\path (xQ0) edge  [bend left]          node {$a$} (xQ1); 
			\path (xQ1) edge [loop below] node {$a,c$} (xQ1); 
			\path (xQ1) edge  [bend left]          node {$b$} (xQ0); 
			
			\node[label]  (yM) [right of=M, node distance=7.2cm, fontscale=1.5]    {$\M_2:\{ \{1\}\}$};
			\node[label] (yP) [below of=yM, node distance=.7cm, fontscale=1.5]  {$\P_2: \kappa(0)=2, \kappa(1)=1$};
			\node[label] (yC) [below of=yP, node distance=.7cm, fontscale=1.5]  {$\C_2: \{ 0 \}$};
			
			\node[initial,state]  (yQ0)   [below left of=yC, node distance=2.2cm] {${0}$};
			\node[state] (yQ1)    [right  of=yQ0]{${1}$};
			
			\path (yQ0) edge [loop below] node {$b,c$} (yQ0); 
			\path (yQ0) edge  [bend left]          node {$a$} (yQ1); 
			\path (yQ1) edge [loop below] node {$a,c$} (yQ1); 
			\path (yQ1) edge  [bend left]          node {$b$} (yQ0);

			\node[label]  (zM) [right of=yM, node distance=7.2cm, fontscale=1.5]    {$\M_3:\{ \{0\}, \{1\}\}$};
			\node[label] (zP) [below of=zM, node distance=.7cm]  { };
			\node[label] (zC) [below of=zP, node distance=.7cm]  { };
			
			\node[initial,state]  (zQ0)   [below left of=zC, node distance=2.2cm] {${0}$};
			\node[state] (zQ1)    [right  of=zQ0]{${1}$};
			
			\path (zQ0) edge [loop below] node {$b,c$} (zQ0); 
			\path (zQ0) edge  [bend left]          node {$a$} (zQ1); 
			\path (zQ1) edge [loop below] node {$a,c$} (zQ1); 
			\path (zQ1) edge  [bend left]          node {$b$} (zQ0);

			\node[label]  (aM) [right of=zM, node distance=7.5cm]  {}; 
			\node[label] (aP) [below of=aM, node distance=.5cm]  { };
			\node[label] (aC) [below of=aP, node distance=.5cm]  { };
			
			\node[label]  (aQ0)   [below left of=aC, node distance=1.8cm] {}; 
			
			\node[label] (pB) [right of=zM, node distance=4.8cm, fontscale=1.5]    {$\B_{\bowtie}:$};
			\node[initial,state,accepting]  (p0)   [right of=aQ0, node distance=1cm] {${0}$};
			\node[state] (p1)    [above right of=p0, node distance=1.8cm]{${1}$};
			\node[state] (p2)    [above left  of=p0, node distance=1.8cm]{${2}$};
			\node[state] (p3)    [below right of=p0, node distance=1.8cm]{${3}$};
			\node[state] (p4)    [below left  of=p0, node distance=1.8cm]{${4}$};
			
			\path (p0) edge           node [right]{$a$} (p1); 
			\path (p1) edge           node [above] {$b$} (p2); 
			\path (p2) edge           node [left] {$a$} (p0); 
			\path (p0) edge           node {$b$} (p3); 
			\path (p3) edge           node {$a$} (p4); 
			\path (p4) edge           node {$b$} (p0); 
			
			\end{tikzpicture}
	}}
	\caption{On the left, examples for non closure of union for $\IM$, $\IP$, and $\IC$. On the right $\B_{\bowtie}$.}
	\label{fig:union-non-closure}
	\label{fig:bowtie}
\end{figure}
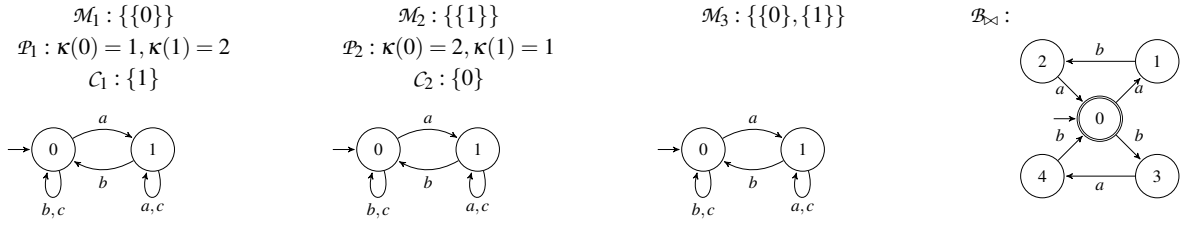


\commentout{
	\begin{table}
		\begin{center}
			\scalebox{0.8}{
				\begin{tabular}{|r||c|c|c|c|c|c|c|c|c|c|}
					\hline 
					Class & $\DB \cap \DC$ & \IB & \IC &  \IP & \IM  & \RB & \RC &  \RP & \RM \\ \hline \hline
					Complementation &  \cmark & \xmark & \xmark &  \cmark & \cmark   & \xmark & \xmark &  \xmark & \xmark\\
					Union &  \cmark & \xmark & \xmark &  \xmark & \xmark    & \xmark & \xmark &  \xmark & \xmark\\ 
					Intersection & \cmark & \xmark & \xmark &  \xmark & \xmark  & \xmark  & \xmark &  \xmark & \xmark\\  \hline
			\end{tabular}} \\ 
		\end{center}
		\caption{Boolean closure properties summary}
		\label{tbl:closure-propoerties}
\end{table}}


\section{Discussion}
\label{sec:discussion}
We have explored properties of the right congruences of regular $\omega$-languages,
characterized when a language has a trivial right congruence, defined classes of
languages that have a fully informative right congruence, and defined an orthogonal property of a language being respective of its right congruence, which is implied by but does not imply the property of being non-counting.
We have shown that there are languages with fully informative right congruences in every class of the infinite Wagner hierarchy, and that this remains true if we consider languages that are also respective of their right congruences.  The (mostly) non-closure results under Boolean operations are not necessarily inimical to learnability.  Our hope is that future research will be able to take advantage of these properties in the search for efficient minimization and learning algorithms for regular $\omega$-languages.



\bibliographystyle{eptcs}
\bibliography{bib}

\end{document}